\documentclass[11pt]{article}
\usepackage{fullpage}
\usepackage{mathrsfs}
\usepackage{float}
\usepackage{booktabs}
\usepackage{amsmath, amsfonts, amssymb, amsthm, amsbsy, amscd, bm, bbm} 
\usepackage{array}
\usepackage{graphicx}
\usepackage[small,bf]{caption}
\usepackage{subcaption}
\usepackage{color}
\usepackage{ifthen}
\usepackage{xspace}
\usepackage{algorithm, algorithmic}
\usepackage[colorlinks,citecolor={black},urlcolor={black},linkcolor={black}]{hyperref}
\usepackage{url}
\usepackage{natbib}
\setcitestyle{authoryear,open={(},close={)}}

\newtheorem{theorem}{Theorem}
\newtheorem{lemma}{Lemma}
\newtheorem{corollary}{Corollary}

\theoremstyle{remark}
\newtheorem{remark}{Remark}

\theoremstyle{remark}
\newtheorem{step}{Step}

\definecolor{wjs}{RGB}{50,200,100}
\definecolor{lbj}{RGB}{200,50,100}

\graphicspath{{fig/}}

\newcommand{\e}[1]{\ifthenelse{\equal{#1}{1}}{\mathrm{e}}{\mathrm{e}^{#1}}}

\newcommand{\R}{\mathbb{R}}

\newcommand{\E}{\mathbb{E}}

\newcommand\iid{i.~i.~d.~}

\newcommand\nulls{\mathscr{N}_0}
\newcommand\T{\top}
\newcommand\ko{\bm{X}_{\textnormal{\tiny KO}}}

\newcommand\NB{\mathrm{NB}}
\newcommand\pr{\mathbb{P}}

\makeatletter
\def\blfootnote{\gdef\@thefnmark{}\@footnotetext}
\makeatother

\begin{document}


\title{Familywise Error Rate Control via Knockoffs}

\author{Lucas Janson \and Weijie Su}
\date{}
\maketitle

{\centering
  \vspace*{-0.3cm} Department of Statistics, Stanford University, Stanford, CA 94305, USA\par\bigskip \date{May 2015}\par
}
\blfootnote{The authors contributed equally to this work.}

\begin{abstract}
We present a novel method for controlling the $k$-familywise
error rate ($k$-FWER) in the linear regression setting using the knockoffs
framework first introduced by Barber and Cand\`{e}s. Our
procedure, which we also refer to as knockoffs, can be applied with any
design matrix with at least as many observations as variables, and
does not require knowing the noise variance. Unlike other multiple
testing procedures which act directly on $p$-values, knockoffs is
specifically tailored to linear regression and implicitly accounts for the statistical
relationships between hypothesis tests of different coefficients. We
prove that knockoffs controls the $k$-FWER exactly in
finite samples and show in simulations that it provides superior power
to alternative procedures over a range of linear regression
problems. We also discuss extensions to controlling other Type I error
rates such as the false exceedance rate, and use it to identify
candidates for mutations conferring drug-resistance in HIV.
\end{abstract}

{\small
{\bf Keywords.}
$k$-familywise error rate; knockoffs; multiple testing; linear regression; Lasso; negative binomial distribution.
}

\section{Introduction}
\label{sec:introduction}
Multiple testing has received increasing attention
with the advent of fields like genetics, technology, and astronomy
which produce very high-dimensional datasets. The increasing number of
hypotheses being simultaneously tested has motivated extensive
research into procedures that maintain control of the familywise errors
that abound when each hypothesis is only tested individually. For
instance, the canonical criterion of the familywise error rate (FWER)
controls the probability of falsely rejecting any of the true null
hypotheses. A number of more modern landmark works have introduced new
Type I error rates that allow for higher power by relaxing the FWER, including the false discovery rate (FDR) \citep{BH1995},
the $k$-FWER \citep{hh1988,lehmann2005}, and the false discovery exceedance (FDX) \citep{genovese2004stochastic, van2004augmentation}. Each one has a different
interpretation, but all control an error rate defined over all
hypotheses being tested, so that conclusions can be drawn by
considering rejected hypotheses together.

Among multiple testing problems, some of most important deal with
finding relationships between variables. Such investigations are often
posed as a linear model
\[
\bm y = \bm X \bm\beta + \bm z,
\]
where $\bm X = [\bm X_1, \ldots, \bm X_p] \in \R^{n \times p}$ is a design matrix,
$\bm\beta \in \R^p$ is a signal vector of interest, and $\bm z \in \R^n$ is the error term. The hypotheses of interest are which variables $\beta_j$, after controlling for all other variables, contribute to
the model, or have nonzero coefficients. With the ability to encode
correlations between variables, linear models capture far more
real-life examples than sequence models. Examples abound particularly
in genetics, where one searches for relationships between parts of the
genome, often in the form of single nucleotide polymorphisms or
expression levels, and continuous variables such as health factors or
drug response. Unfortunately, due to the dependence among the
variables in the linear model, their respective tests do not in general
exhibit any of the simple dependence structures, such as independence
or positive dependence,
that are required for many of the most powerful existing procedures.

In this work we focus on controlling the $k$-FWER, the probability of making at least
$k$ false discoveries, in the context of 
linear models. Our method uses the framework of
knockoffs introduced by \citet{knockoff}. The idea of knockoffs is to
carefully construct artificial variables that serve as controls for
the original variables. Barber and Cand\`es show that these controls
are easy to construct and can be used to automatically account for
variable dependence to provide finite-sample FDR control
for general design matrices without knowledge of the noise
variance. Controlling the FDR can be highly
desirable in a high-power setting, but results can be hard to
interpret when few discoveries are made, as the realized false discovery
proportion may be highly variable. The $k$-FWER,
which in the case of $k=1$ reduces to the standard FWER, always has a
clear interpretation by explicitly bounding the
probability of $k$ or more false discoveries, making it a useful
criterion in all settings, as evidenced by its wide acceptance in the
scientific community. The $k$-FWER also provides a fundamental building
block to other Type I error rates, such as the FDX and Per Family
Error Rate (PFER), as we will discuss in
Section~\ref{sec:contr-other-error}. We leverage the attractive
features of the knockoffs framework to construct a novel procedure for controlling the
$k$-FWER that implicitly accounts for the exact dependence structure
in linear regression problems. In particular, we prove finite-sample
$k$-FWER control for general design matrices without
any knowledge of the noise variance, and show in simulations that the
power can be substantially greater than state-of-the-art alternatives.

Much previous work has studied controlling the $k$-FWER under
varying assumptions on the statistical dependence among the hypothesis
test statistics or $p$-values. The bulk of such work has dealt with
procedures that act directly on the $p$-values. When there are more
observations than variables and the noise is \iid Gaussian, ordinary least squares regression
generates dependent $t$-statistics for all variables, allowing those
procedures that can account for the dependence structure to be applied
to the associated $p$-values. Unfortunately, the joint distribution of such $p$-values does not generally satisfy popular dependence assumptions
such as positive regression dependence on subset \citep{benjamini2001}
or multivariate total positivity \citep{Karlin1980}. Furthermore,
many of the procedures that can account for general dependence
structures do so nonparametrically through resampling. However,
resampling procedures tend to require extra assumptions such as
subset-pivotality \citep{westfall1989} which do not hold in general in
the regression setting, or only provide exact control
asymptotically \citep{romano2007}. We mention here some work on controlling 
the $k$-FWER in finite samples and refer the reader to \cite{guo2014}
for a more thorough review. The most popular methods for FWER control are the Bonferroni \citep{Dunn1961} and Holm's \citep{Holm1979}
procedures, neither of which require assumptions on the dependence
among $p$-values. Under independence, the Bonferroni procedure can be
improved using the \v{S}id\'{a}k correction \citep{Sidak1967}, or one can employ Hochberg's
step-up procedure \citep{Hochberg1988}. In \citet{lehmann2005},
step-down procedures generalizing Bonferroni and Holm's procedures are
presented, while \citet{romano2007} introduce a generic step-down
procedure, all for controlling the $k$-FWER. \citet{romano2006} also present step-up procedures for
controlling the $k$-FWER under arbitrary unknown dependence. 

To avoid confusion, we point out that the recent work of
\cite{lockhart2014} provides $p$-values for coefficients in a
linear model, however they deal with a different notion of
a null hypothesis than used here. In their framework, the null
hypotheses are defined sequentially with
respect to a growing model, wherein each time the model size is
increased by one, the null hypothesis is that the new variable is
uncorrelated with the response, conditional on only the variables
already included in the model.

The remainder of the paper is structured as
follows. Section~\ref{sec:prel-knock} introduces notation and gives a
short introduction to the knockoffs framework. Section~\ref{sec:familyw-error-contr}
describes the knockoffs procedure for control of the
$k$-FWER and proves this control along with tail
bounds. Section~\ref{sec:contr-other-error} provides a brief
discussion of how the procedure can be used to control the PFER and
FDX. Section~\ref{sec:comp-with-other} compares our procedure to
state-of-the-art alternatives from the literature, both in terms of
practical considerations and power, in a series of
simulations. Section~\ref{sec:real-data-experiment} demonstrates an
implementation on a real dataset from genetics, and
Section~\ref{sec:discussion} concludes with discussion and directions
for future research.


\section{Preliminaries for knockoffs}
\label{sec:prel-knock}
In this section, we introduce the knockoffs machinery of
\citet{knockoff} at a minimal level to be sufficient for our
exposition of $k$-FWER control. This material is largely
borrowed from the reference \cite{knockoff}. In referring to the
knockoffs framework, we always assume that the number of observations
$n$ is at least the number of variables $p$, the design matrix $\bm X$ has
full rank so that the Gram matrix $\bm X^{\T} \bm X$ is invertible, and the
noise term $\bm z$ has independent Gaussian entries. We would like to briefly
emphasize here that $n \ge p$ is necessary for the multiple hypothesis
testing problem to even be well-defined. For any linear regression
problem, the ``true'' coefficient vector is only statistically
well-defined modulo addition with any vector in the null space of the
design matrix. If $p > n$, then the design matrix has a nontrivial
null space, thus allowing zeros and nonzeros in the coefficient vector
to arise and disappear, changing the fundamental values of the null
hypotheses, without changing the data-generating process at
all. Except for this non-degeneracy assumption, the knockoffs
machinery works for general designs $\bm X$ and does not even require
knowledge of noise variance $\sigma^2$.

To start with, again, consider the linear model
\[
\bm y = \bm X \bm \beta + \bm z,
\]
where the noise vector $\bm z$ has independent $\mathcal{N}(0, \sigma^2)$ entries, and each column of $\bm X$ has been normalized to have unit $\ell_2$-norm, that is, $\|\bm X_j\| = 1$ for all $1 \le j \le p$. The first step of this method is to construct the knockoff design, denoted as $\widetilde{\bm X} \in \R^{n \times p}$, that obeys
\begin{equation}\label{eq:ko_symm}
\widetilde{\bm X}^{\T}\widetilde{\bm X} = \bm X^\T \bm X, \quad \bm X^\T \widetilde{\bm X} = \bm X^\T \bm X - \operatorname{Diag}(\bm s),
\end{equation}
where $\bm s \in \R^p$ has nonnegative entries and the superscript $\T$ denotes matrix transpose hereafter. There are multiple ways to construct this knockoff
design. The first equality forces $\widetilde{\bm X}$ to have the same
correlation structure among its columns as $\bm X$. In the ideal case
of $n \ge 2p$, it can be guaranteed that the $2p$ column vectors of
$\bm X$ and $\widetilde{\bm X}$ are jointly linearly independent. By the second equality, for every $1 \le j \le p$, the original variable $\bm X_j$ and the knockoff counterpart $\widetilde{\bm X}_j$ have the same correlation with all the other $2p - 2$ variables, namely, $\bm X_i, \widetilde{\bm X}_i$ for $i \ne j$. At a high level, we can view the knockoff design as a control group as compared to the original design $\bm X$, which is treated as the case group.

Denote by $\ko = [\bm X, \widetilde{\bm X}] \in \R^{n \times 2p}$ the
concatenation of the original design and the knockoff design. With
$\ko$ in hand, the next step is to generate statistics for each
variable. One way to do so, suggested in \citet{knockoff}, is by fitting the entire Lasso regularization path on the augmented design,
\begin{equation}\label{eq:Lasso}
\widehat{\bm\beta}(\lambda) = \underset{\bm b \in \R^{2p}}{\operatorname{argmin}} \quad \frac12 \|\bm y - \ko \bm b\|^2 + \lambda\|\bm b\|_1,
\end{equation}
and letting $Z_j $ be the first $\lambda$ such that $\widehat\beta_j$ is nonzero. Formally,
\[
Z_j = \sup\{\lambda: \widehat\beta_j(\lambda) \ne 0\}.
\]
As pointed in the reference paper, many alternative statistics,
including some based on least-squares, least angle regression
\citep{lars} and sorted-$\ell_1$-penalized estimation \citep{slope,
  slopeminimax}, can be used instead as long as they obey the
sufficiency and antisymmetry properties defined therein. Defining
$\widetilde Z_j$ analogously for each knockoff variable $\widetilde{\bm X}_j$, the knockoff statistics (using slightly different notation
than in the original paper) are
\[
W_j = \max\{Z_j, \widetilde Z_j \}, \quad \chi_j = \operatorname{sgn}(Z_j - \widetilde{Z}_j),
\]
where $\operatorname{sgn}(x) = -1, 0, 1$ if $x < 0, x= 0, x > 0$,
respectively. The following result, due to \citet{knockoff},
characterizes the joint distribution of the null $\chi_j$. We say $j$
is a true null when $\beta_j = 0$ and a false null otherwise.
\begin{lemma}[\citet{knockoff}]\label{lm:knockoff_key}
Conditional on all $W_j$ and all false null $\chi_j$, all true null
$\chi_j$ are jointly independent and uniformly distributed on $\{-1,1\}$.
\end{lemma}

This simple lemma is very helpful in proving $k$-FWER
control. Its proof follows from the symmetry between $\bm X_j$ and
$\widetilde{\bm X}_j$ if $\beta_j = 0$, which is provided by the
construction \eqref{eq:ko_symm}. The lemma shows that $\chi_j$ can be
interpreted as a one-bit $p$-value, in the sense that it has equal chance to
take $1$ or $-1$ if $\beta_j = 0$. In fact when $\beta_j = 0$, the
knockoff symmetry characterized in \eqref{eq:ko_symm} introduces
exchangeability between $\bm X_j$ and its knockoff counterpart $\widetilde{\bm X}_j$ in the Lasso path \eqref{eq:Lasso}. Hence, $\bm X_j$ and $\widetilde{\bm X}_j$
are equally likely to enter the Lasso path first. Conversely, if
$\beta_j \neq 0$, then $\bm X_j$ is likely to enter
before $\widetilde{\bm X}_j$ so that $\chi_j=1$. Thus a large $W_j$ and a
positive $\chi_j$ provide evidence against the $j$th null hypothesis $H_{0, j}: \beta_j = 0$.


\section{$k$-familywise error rate control}
\label{sec:familyw-error-contr}
Inspired by the interpretation of the statistics $W_j$ and $\chi_j$,
it is reasonable to reject hypotheses with positive signs $\chi_j$ and
large $W_j$. Parameterized by a positive integer $v$, the knockoffs procedure
for controlling the $k$-FWER is as follows.
\begin{step}
Denote by $W_{\rho(1)} \ge W_{\rho(2)} \cdots \ge W_{\rho(p)}$ the order statistics of $\bm W$, where $\rho(1), \ldots, \rho(p)$ is a permutation of $1, \ldots, p$.
\end{step}
\begin{step}
Let $j^\star$ be the index of the $v$th $-1$ in the sequence $\chi_{\rho(1)}, \ldots, \chi_{\rho(p)}$. If fewer than $v$ negatives appear, set $j^\star = p$.
\end{step}
\begin{step}
Reject all the null hypotheses $H_{0, \rho(j)}$ whenever $j \le j^\star$ and $\chi_j = +1$.
\end{step}
More compactly, define the threshold
\begin{equation}\nonumber
T_v = \sup \Big\{ t > 0: \#\{j: W_j \ge  t, \chi_j = -1\} = v \Big\},
\end{equation}
with the usual convention that $\sup \emptyset = -\infty$. The
multiplicity of $W_j$ is not accounted for since all $W_j$ are unique
with probability 1. Then, the knockoffs procedure rejects all $H_{0, j}$ with $W_j \ge T_v$ and $\chi_j = +1$.

Before characterizing the distribution of false discoveries made by
the knockoffs procedure, we define some notation. Let $\nulls = \{1 \le
j \le p: \beta_j = 0\}$ be the set of true null hypotheses and $\NB(m,
q)$ denote a negative binomial random variable, which counts the
number of successes before the $m$th failure in a sequence
of independent Bernoulli trials with success probability $q$.

\begin{lemma}\label{lm:negative_binomial}
For any integer $v \ge 1$, the false discovery number
\[
V = \ \#\{j \in \nulls: W_j \ge T_v \mbox{ and } \chi_j = +1 \}
\]
is stochastically dominated by $\NB(v, 1/2)$.
\end{lemma}

\begin{proof}[Proof of Lemma~\ref{lm:negative_binomial}]
First, we prove this lemma in the case where $\nulls = \{1, \ldots,
p\}$, that is, $\beta_j = 0$ for all $j$. Conditional on all $W_j$,
Lemma \ref{lm:knockoff_key} concludes that $\chi_{\rho(1)}, \ldots,
\chi_{\rho(p)}$ are independent and each takes $+1$ and $-1$,
respectively, with probability $1/2$. Note that the permutation $\rho$
is deterministic conditional on the $W_j$. Recognizing that $V$ is the
number of positive $\chi_j$ before the $v$th negative or the
$p$th trial happens, whichever comes first, we see that $V$
is an early stopped negative binomial random variable. In the general case, false
null $\chi_j$ will insert $-1$'s into the process on the nulls,
causing it to stop no later than when $\nulls = \{1, \ldots,
p\}$. Therefore, $V$ is always stochastically dominated by $\NB(v, 1/2)$.
\end{proof}

The stochastic upper bound in Lemma~\ref{lm:negative_binomial} is
tight in the following sense. The distribution of $V$ can be made
arbitrarily close to $\NB(v, 1/2)$ under the global null by taking $p \gg v$, as in this case at least $v$ negative $\chi_j$ will
appear in the sequence with high probability. Next we present the main
result, which is immediate from Lemma~\ref{lm:negative_binomial} and
the negative binomial cumulative distribution function.

\begin{theorem}\label{thm:tail}
For any integer $k \geq 1$ and significance $0 < \alpha < 1$, let $v$ to be the largest integer satisfying
\begin{equation}\label{eq:solve_for_v}
\sum_{i = k}^{\infty}2^{-i-v}{i + v -1\choose{i}} \le \alpha.
\end{equation}
Then the knockoffs procedure with parameter $v$ controls the $k$-FWER at level $\alpha$, that is, $\pr(V \geq k) \le \alpha$.
\end{theorem}


As a concrete example, taking $v = 4$ would provide $10$-FWER control at level $0.05$. As one may observe from
\eqref{eq:solve_for_v}, the integer $v$ as a function of the level
$\alpha$ cannot be continuous. Consequently, $\pr(V \geq k)$ is in
general lower than the target level $\alpha$. In particular, for
$\alpha \le 1/2^k$ no positive integer $v$ satisfies
\eqref{eq:solve_for_v}, so the naive procedure must reject
nothing. This matter can be easily resolved by randomization of $v$,
as we will show in Remark~\ref{rem:power}.


To better understand the knockoffs procedure, we may want to know how
many false rejections are made when the $k$-FWER is not
controlled. To this end, the following result bounds the tail
probability of $V$, or the probability of making many more rejections
than expected.
\begin{corollary}\label{cor:chernoff}
For arbitrary $a > 0$, the error rate of the knockoffs procedure with parameter $v$ obeys
\[
\pr(V \ge (1+a)v ) \le \theta(a)^v,
\]
where $\theta(a) = \frac{(a+2)^{a+2}}{2^{a+2}(a+1)^{a+1}} < 1$.
\end{corollary}
\begin{proof}[Proof of Corollary \ref{cor:chernoff}]
By Lemma \ref{lm:negative_binomial}, it suffices to prove the inequality when $V$ is distributed as $\mathrm{NB}(v, 1/2)$. For any positive number $\eta < \log 2$, from the Markov inequality we get
\begin{equation}\nonumber
\pr(V \ge k) \le  \frac{\E (\e{\eta V} )}{\e{(1+a)\eta v}} = \frac{1}{(2 - \e{\eta})^v\e{(1+a) \eta v}}.
\end{equation}
The desired bound follows from taking $\eta = \log (2 + 2a) - \log(2 + a)$.
\end{proof}

\begin{remark}[Power Improvement]
\label{rem:power}
As mentioned earlier, the knockoffs procedure suffers from a
discretization problem, especially for small $k$, but this can be
remedied by randomization as follows. For any desired level $\alpha
\in (0, 1)$, there must exist an integer $v \ge 0$ such that
\[
\pr_v(V \ge k) \le \alpha \le \pr_{v+1}(V \ge k),
\]
where the subscript $v$ or $v+1$ emphasizes the parameter of the
knockoffs procedure. We can devise a mixture procedure that
obeys exactly $\pr(V \ge k) = \alpha$ by putting weights $\omega$ and
$1 - \omega$, respectively, on the knockoffs procedures with parameters
$v$ and $v+1$, where
\[
\omega = \frac{\pr_{v+1}(V \ge k)-\alpha}{\pr_{v+1}(V
  \ge k) - \pr_v(V \ge k)}.
\]

Furthermore, as with any procedure controlling the $k$-FWER,
power can always be improved without affecting the $k$-FWER by always making at least $k-1$ rejections. In the case of
knockoffs, if we were going to make fewer than $k-1$ rejections, we can simply
continue rejecting the indices with the largest $W_j$ and positive
$\chi_j$ until there are $k-1$. The benefit of this modification
depends on the ordering of the hypotheses induced by $W_j$. 
\end{remark}

\section{Controlling other error rates}
\label{sec:contr-other-error}

This paper has been about controlling the $k$-FWER, but the procedure
introduced can be used to control other Type I error rates as well,
namely the PFER and the FDX.

Originally proposed by John Tukey in an unpublished work in 1953, the
PFER is defined as $\E (V)$, or in words, the
expected number of false discoveries. The control of this error rate
under general $p$-value dependence has not received as much attention in
the literature as other error rates, although both \cite{gordon2007}
and \cite{meng2014} have discussed using the Bonferroni procedure for
this purpose. Lemma~\ref{lm:negative_binomial} shows that the knockoffs
procedure for controlling the $k$-FWER also controls
the PFER at level $v$, as $\E (V) \le \E \, \NB(v, 1/2)
= \frac{1/2}{1 - 1/2}v = v$.

The FDX, also known as the $\gamma$-false discovery proportion, tail
probability for the proportion of false positives, or false
discovery excessive probability, is the probability that the FDP exceeds a
specified bound $\gamma$. It can be viewed as a more
stringent form of the FDR, and has received much
attention recently; see, for example \citet{guo2014}. A number of authors have noticed its
intimate connection with the $k$-FWER, and many of the most successful
FDX-controlling procedures in the literature can be posed as
meta-procedures applied to a family of $k$-FWER-controlling
procedures
\citep{van2004augmentation,genovese2004stochastic,romano2007}. We
briefly review three such meta-procedures, any one of which could be
combined with the knockoffs procedure introduced here, and defer further
investigation to future work.

In \citet{van2004augmentation}, the authors introduced a simple and intuitive
procedure which augments any FWER-controlling procedure to control the
FDX. This procedure was generalized to any $k$-FWER-controlling
procedure in \cite{genovese2006exceedance}. Once the
$k$-FWER-controlling procedure makes $R$ rejections, then if
$(k-1)/R>\gamma$, the augmentation
procedure makes no rejections, but if $(k-1)/R\le\gamma$, $r$ more
rejections can be made, where $r$ satisfies
$(k-1+r)/(R+r)\le\gamma$. This augmentation procedure controls the FDX
exactly when the underlying $k$-FWER-controlling procedure also
provides exact control.

\citet{genovese2004stochastic} proposed a test-inversion procedure
for FDX control, similar to the closure principle of \citet{marcus1976} for
FWER control, which was then investigated further in
\citet{genovese2006exceedance}. The inversion procedure runs global null
hypothesis tests on every subset of hypotheses, and then finds the
largest subset $S$ whose maximal
intersection with any subset for which the global null was not
rejected is at most $\gamma |S|$. Note that any $k$-FWER-controlling
procedure is also trivially a test of the global null hypothesis,
rejecting whenever $k$ or more rejections are made. Rejecting $S$
from the inversion procedure controls the FDX exactly, and although
in general it takes exponential time, for some global tests it can be
run in polynomial time \citep{genovese2004stochastic}.

Given a procedure that can control the
$k$-FWER for any $k \ge 1$, \citet{romano2007}
propose a heuristic that aims to control the FDX. In short,
given a prescribed level $\gamma$ and significance $\alpha$, both
between 0 and 1, this heuristic uses a $k$-FWER-controlling procedure
to make rejections for increasing $k$ until just before the number of
rejections goes above $k/\gamma-1$. Explicitly, let $R_k$ be the
number of rejections made by a procedure controlling the $k$-FWER. Then the
Romano--Wolf heuristic defines $\hat k$ as the smallest $k$ such that
$R_k < k/\gamma-1$ and makes rejections as if controlling the
$\hat k$-FWER. Although not rigorous due to its adaptivity in $\hat
k$, under some dependence assumptions, the Romano--Wolf heuristic is
shown to enjoy finite sample or asymptotic FDX control for
step-down procedures \citep{guo2007generalized, delattre2013new}.


\section{Comparison with other procedures}
\label{sec:comp-with-other}
As mentioned in the introduction, the structure and dependence between
coefficients in linear regression preclude the use of many existing
procedures. The state-of-the-art procedures that can be found in existing
literature and provide exact finite-sample control of the $k$-FWER in
linear regression are:
\begin{itemize}
\item[(a)] the generic step-down procedure of \cite{romano2007} applied to
  the least-squares $p$-values
\item[(b)] the step-up procedure of \cite{romano2006} applied
  to the least-squares $p$-values
\item[(c)] the adaptation of Holm's procedure to $k$-FWER applied to the least-squares $p$-values \citep{lehmann2005}
\item[(d)] for 1-FWER, the Lasso pathwise testing
  procedure of \cite{lee2013exact}
\item[(e)] also for the 1-FWER, the closure of any
  global hypothesis testing procedure, such as the $\chi^2$ test,
  that can be applied to $p$-values with any known dependence, applied
  to the least-squares $p$-values
\end{itemize}

Procedure (d) requires the user to know
$\sigma^2$ exactly, and both (d) and (e) take computation that is exponential in
the dimension, $p$, making them infeasible to use for problems of even
moderate size. As a result, we only compare our procedure to (a), (b), and
(c). It should be noted that the problem dimensions we considered
in simulations were still limited by procedure (b), whose computation
time is
$O(p^{k-1})$, since each threshold is computed as a maximum over subsets
of size $k-1$ from a superset of size up to $p$. There are also works
that obtain asymptotic control of the FWER under some assumptions on the distribution of the design matrix
(see, for example, \cite{chernozhukov2013, javanmard2014}). As knockoffs applies
under no assumptions on the design matrix and the error rates are
controlled exactly, we do not compare to such works here.

In each of the following simulations, we performed many independent
experiments to gauge how the performance of knockoffs,
both in absolute terms and relative to previous methods, depends on
correlation in the columns of $\bm X$, the sparsity of $\bm\beta$, and the
signal to noise ratio. In each experiment, $\bm X$ is generated by
normalizing the columns of a multivariate Gaussian matrix with
independent and identically distributed rows, and $\bm\beta$ is generated
by setting a pre-specified number of entries to zero, and setting the
rest to the same nonzero magnitude, which is also prespecified. The
following experiments are all performed in the sparse setting, as that
is what the canonical statistics $\bm W$ that use the Lasso are best-suited
for. However, nothing about the knockoffs framework to control
any Type I error rate is particularly tied to sparsity, and it is of
continuing interest to find different statistics $\bm W$ that achieve high
power in all manner of settings. In all the following simulations,
$n=1000$, $p=450$, $\sigma^2=25$, we control the 5-FWER at the 5\% level, and we apply the modifications in Remark~\ref{rem:power}. The step-up procedure is implemented using the critical
values suggested in \cite{romano2006}, namely their Equation (13). For a sake of reproducibility, the code to generate these figures is available at \url{http://wjsu.web.stanford.edu/code.html}.

Our first experiment took $\bm\beta$ to have 10 nonzero elements, all
with magnitude 10, and varied the pairwise correlation between the
columns of $\bm X$ from 0 to 0.5. 
\begin{figure}[ht!]
\centering
\hfill
\begin{subfigure}[b]{0.49\textwidth}
\centering
 \includegraphics[width=\textwidth]{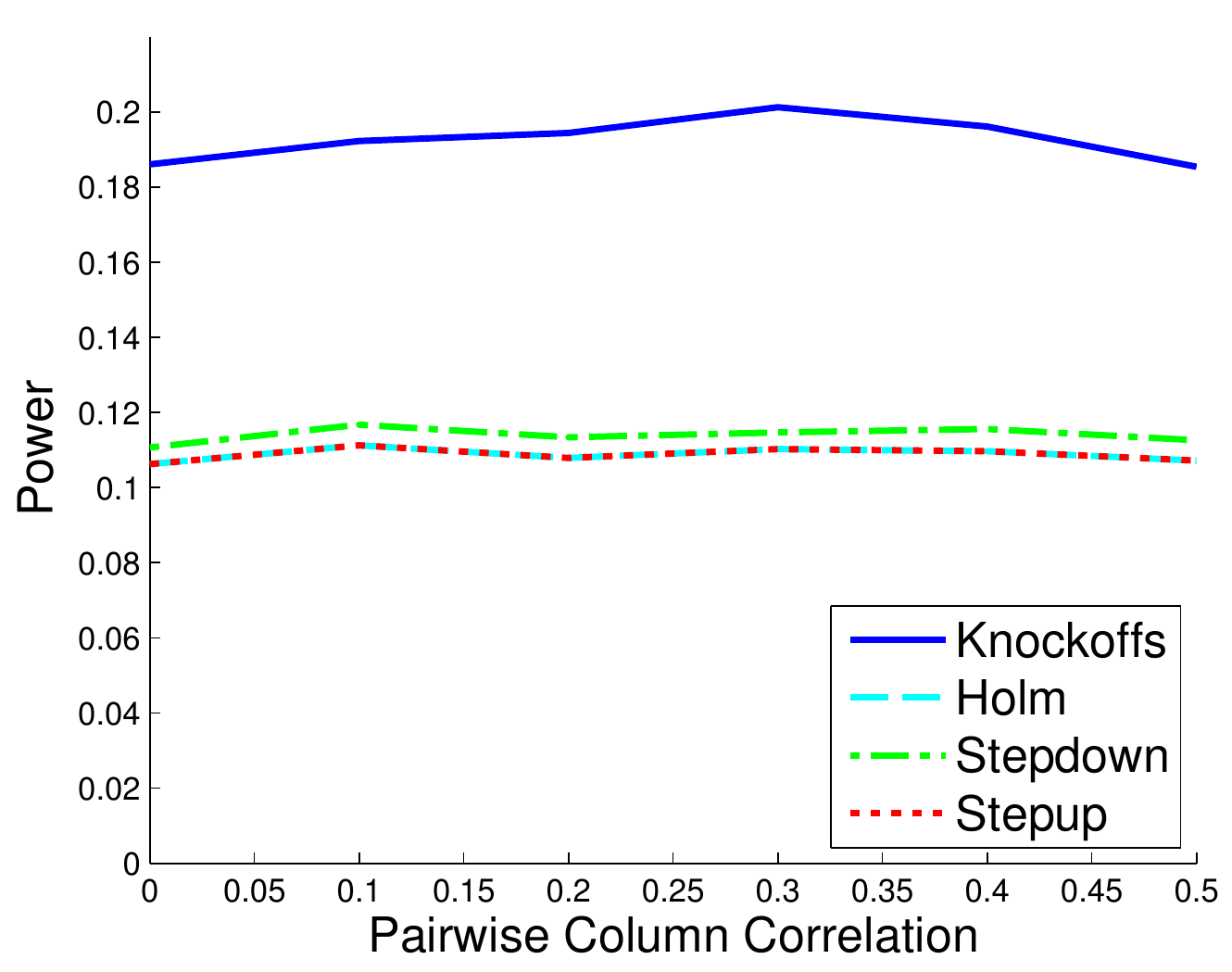}
\label{rhopower}
\end{subfigure}
\begin{subfigure}[b]{0.49\textwidth}
\centering
 \includegraphics[width=\textwidth]{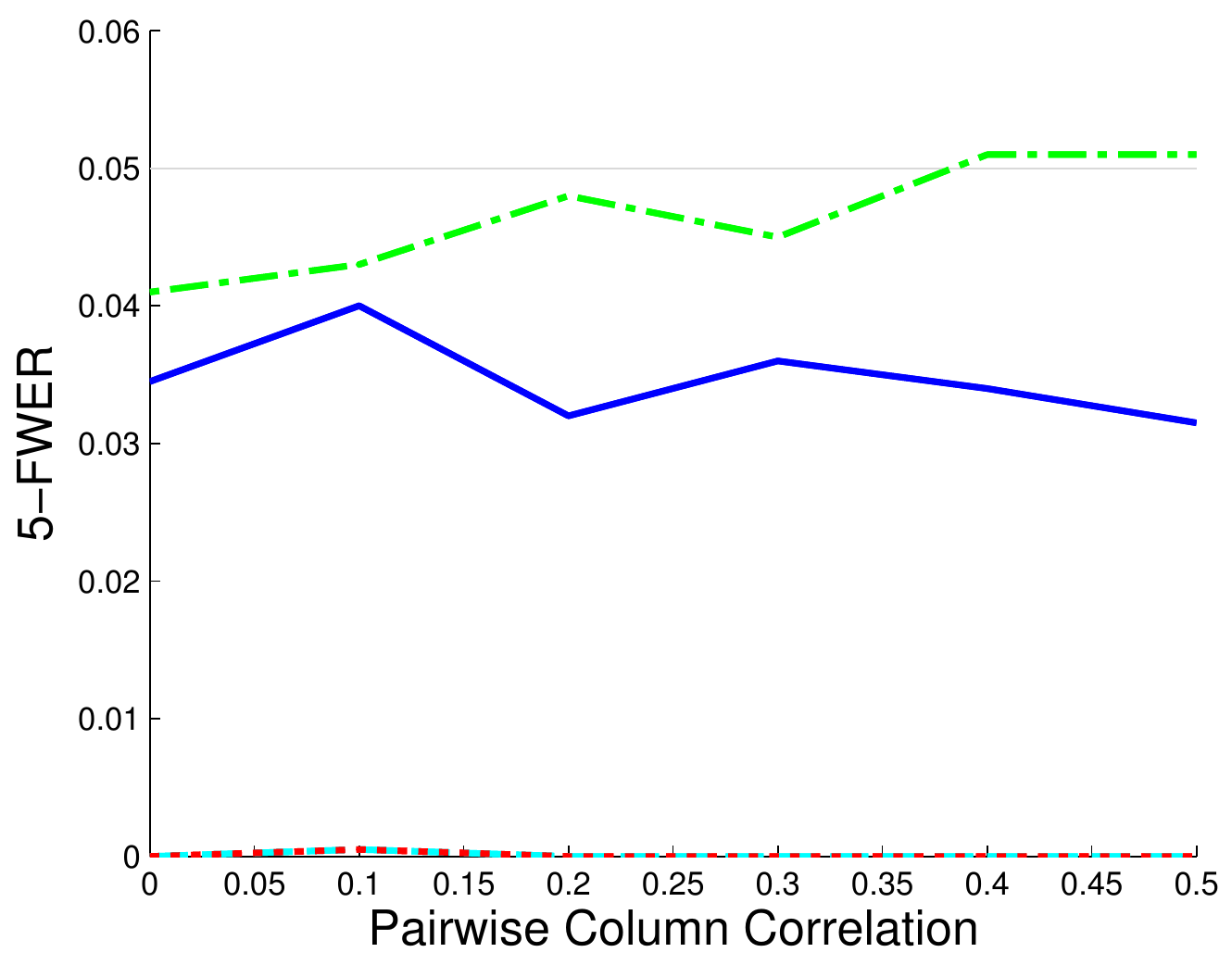}
\label{rhofwer}
\end{subfigure}
\hfill
  \caption{Comparison of Holm's procedure, generic step-down
    procedure, step-up procedure,
    and knockoffs for controlling the 5-FWER at the 5\% significance
    level. As functions of the column correlation of the design
    matrix, the procedures' powers are shown in (a), while
    the 5-FWER is given in (b), with the grey line denoting the
    nominal level of 5\%. The curves for Holm and step-up lie on top
    of one another. Each point is an average over 2000 simulations.}
\label{fig:rho}
\end{figure}
Figure~\ref{fig:rho} shows the power of the knockoff procedure nearly
doubling that of all alternative procedures. The power and
5-FWER of all four procedures is largely unaffected
by the correlation in the columns of $\bm X$.

Our second experiment generated columns for $\bm X$ independently, and
varied the sparsity of $\bm\beta$, with each nonzero coefficient having magnitude 10.
\begin{figure}[ht!]
\centering
\hfill
\begin{subfigure}[b]{0.49\textwidth}
\centering
 \includegraphics[width=\textwidth]{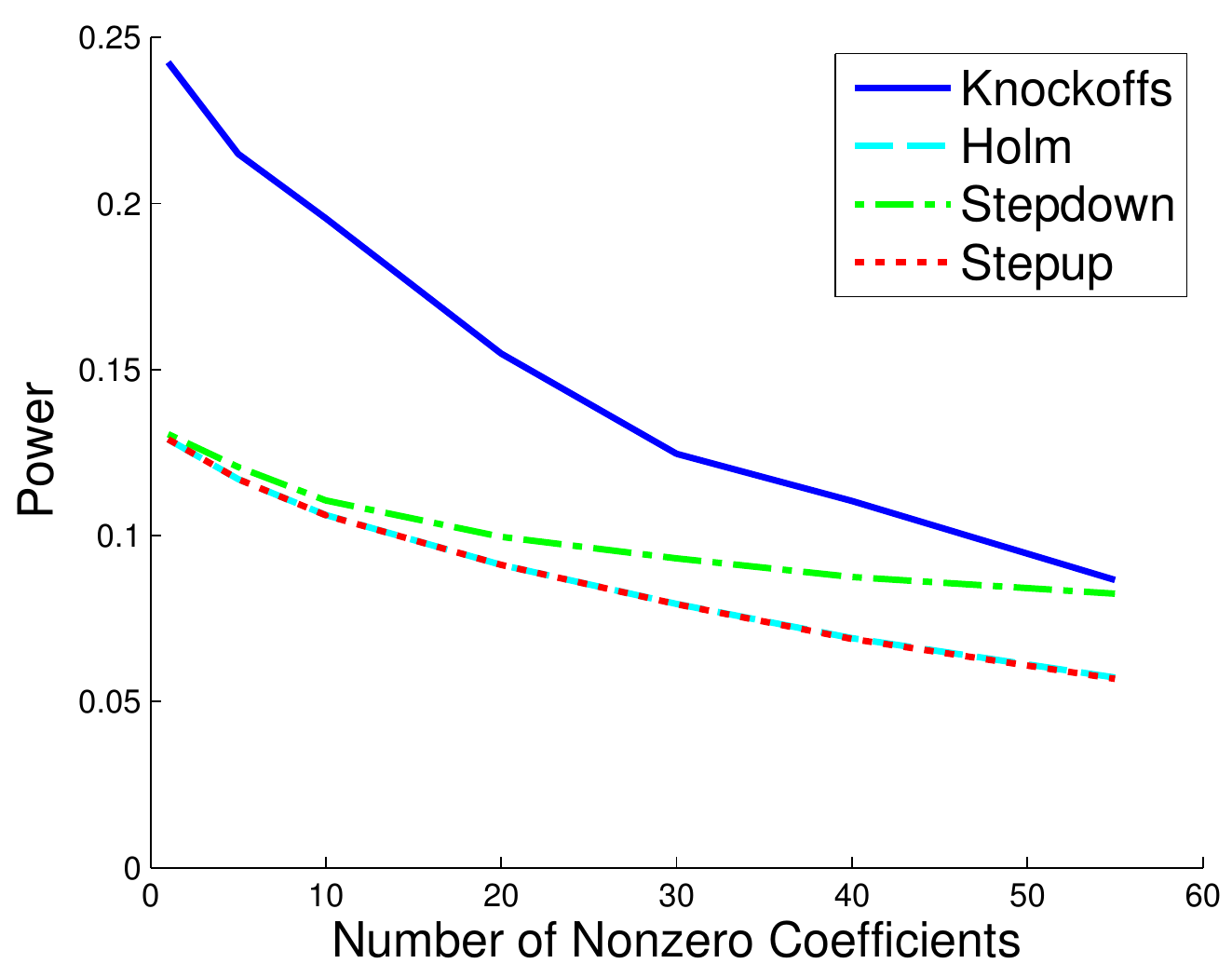}
\label{nnzpower}
\end{subfigure}
\begin{subfigure}[b]{0.49\textwidth}
\centering
 \includegraphics[width=\textwidth]{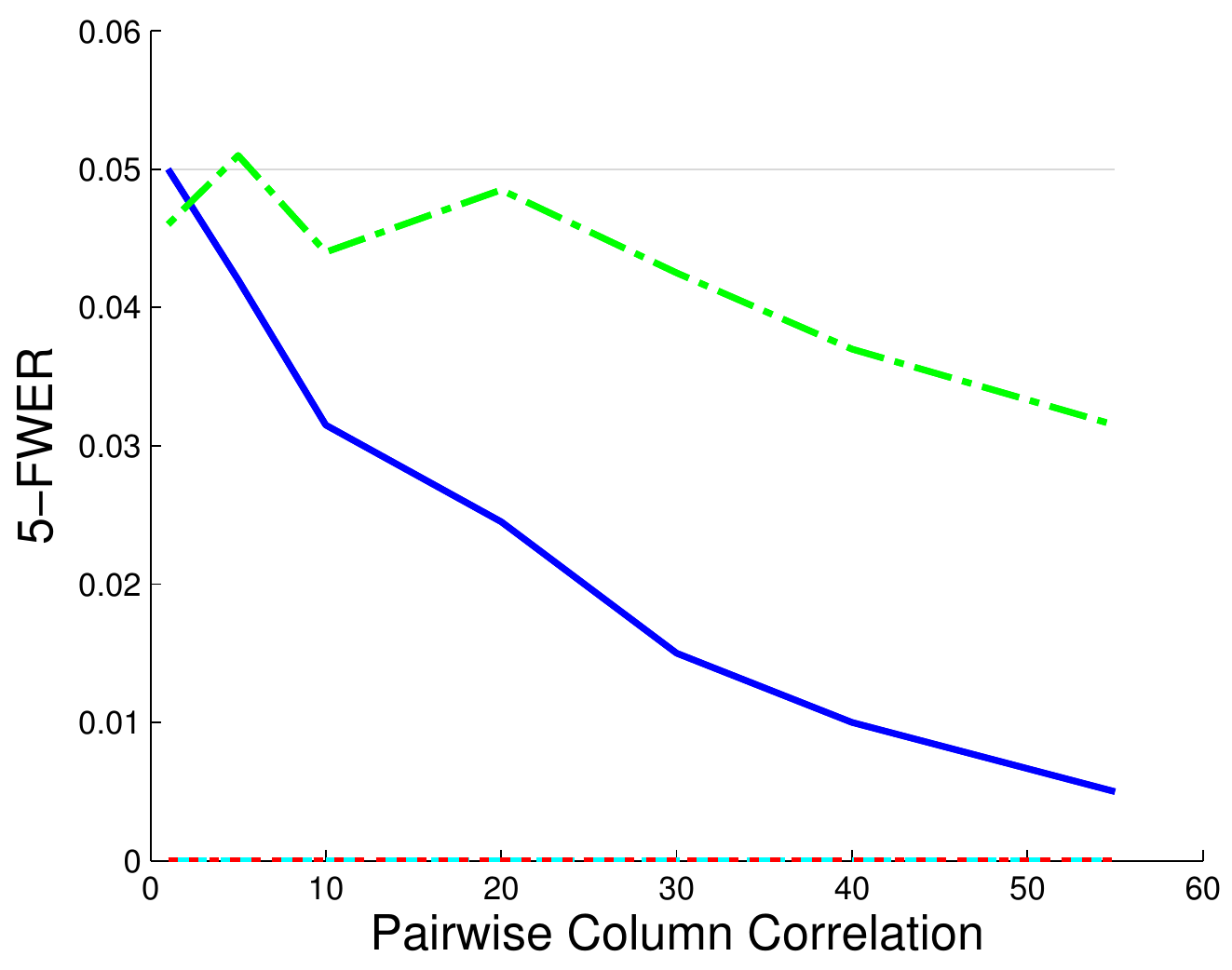}
\label{nnzfwer}
\end{subfigure}
\hfill
  \caption{Comparison of Holm's procedure, generic step-down
    procedure, step-up procedure, and knockoffs for controlling the
    5-FWER at the 5\% significance
    level. As functions of the number of nonzero coefficients, the
    procedures' powers are shown in (a), while
    the 5-FWER is given in (b), with the grey line denoting the
    nominal level of 5\%. The curves for Holm and step-up lie on top
    of one another. Each point is an average over 2000 simulations.}
\label{fig:nnz}
\end{figure}
Figure~\ref{fig:nnz} shows the power of the knockoff procedure approximately
doubling that of all alternative procedures in the sparsest regime
and gradually losing its advantage as the sparsity approaches
10\%. The 5-FWER of the knockoffs and step-down decrease as the
coefficient vector becomes less sparse, with that of knockoffs becoming
conservative especially quickly.

Our third experiment generated independent columns for $\bm X$, used
$\bm\beta$ with 10 nonzero entries, and varied the magnitude of the
nonzero entries on a logarithmic scale.

\begin{figure}[ht!]
\centering
\hfill
\begin{subfigure}[b]{0.49\textwidth}
\centering
 \includegraphics[width=\textwidth]{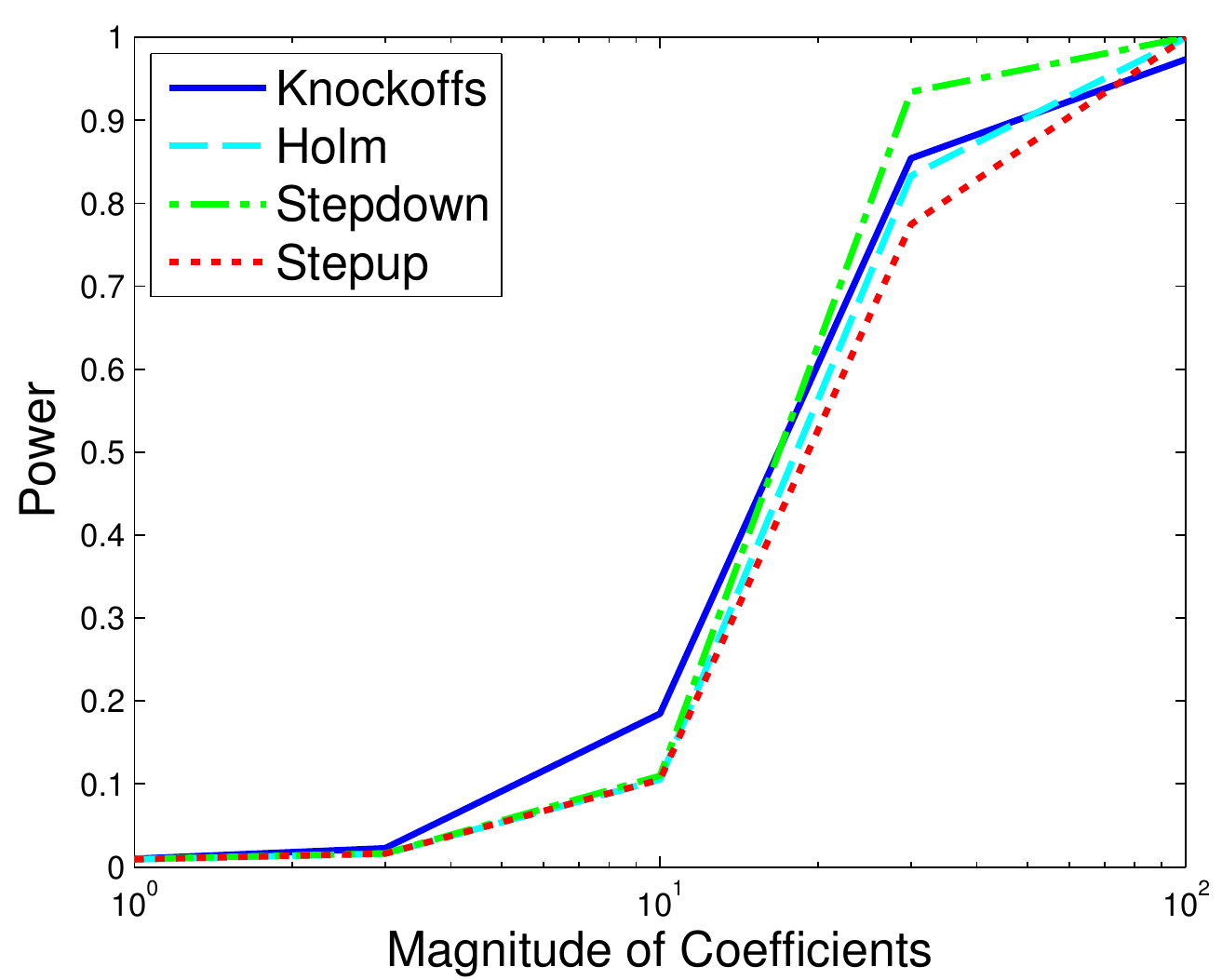}
\label{magpower}
\end{subfigure}
\begin{subfigure}[b]{0.49\textwidth}
\centering
 \includegraphics[width=\textwidth]{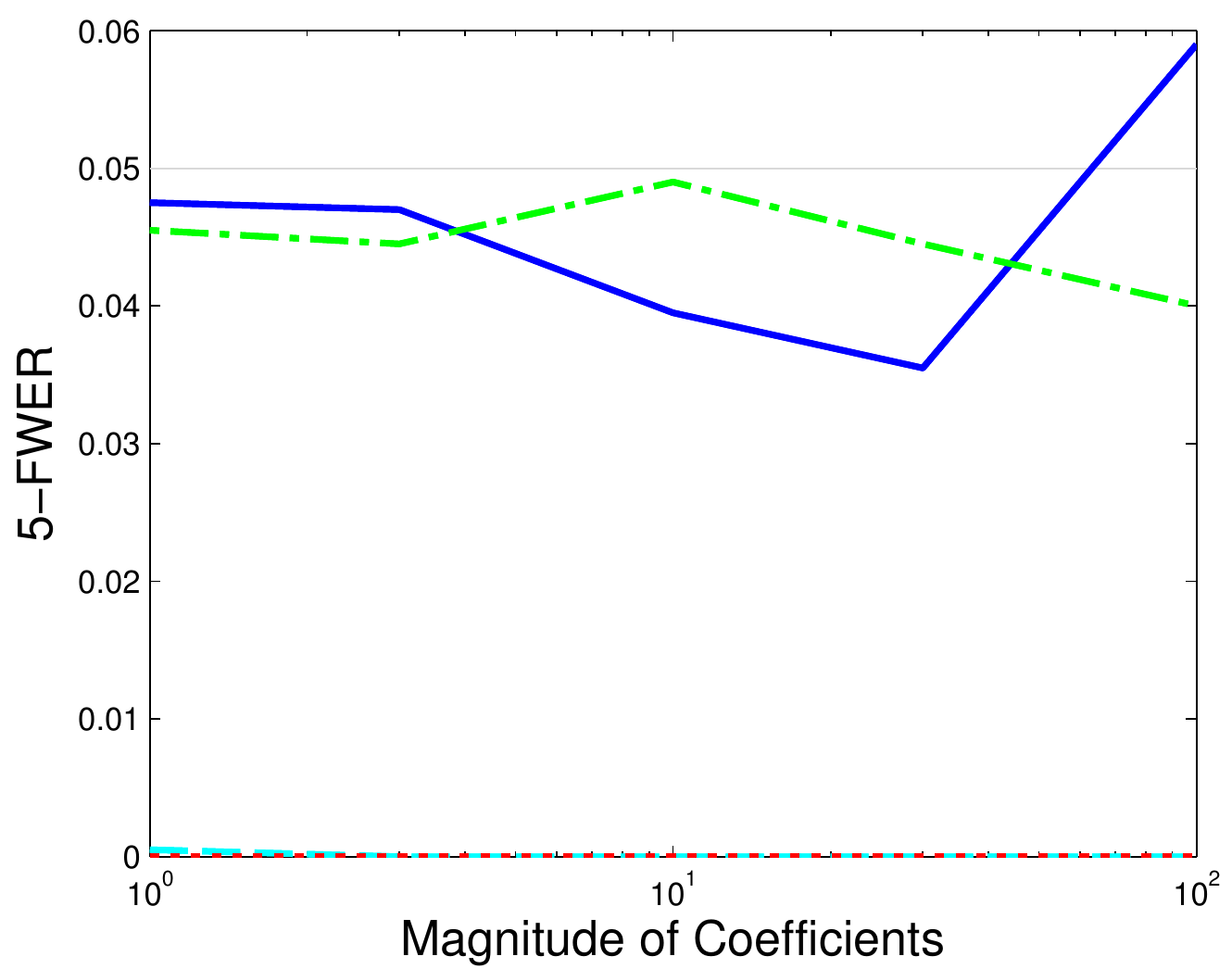}
\label{magfwer}
\end{subfigure}
\hfill
\caption{Comparison of Holm's procedure, generic step-down
    procedure, step-up procedure and knockoffs for controlling the
    5-FWER at the 5\% significance
    level. As functions of the magnitude of the nonzero coefficients, the
    procedures' powers are shown in (a), while
    the 5-FWER is given in (b), with the grey line denoting the
    nominal level of 5\%. Each point is an average over 2000 simulations.}
\label{fig:mag}
\end{figure}

Figure~\ref{fig:mag} shows the power of the knockoff procedure above
all alternative procedures in the low- to middle-power regimes, while
it actually has slightly less power in the very high-power regime,
corresponding to a signal-to-noise ratio $\|\bm\beta\|^2/\sigma^2 >
350$. This reversal can be explained by the fact that with non-orthogonal columns and a not-extremely-sparse $\bm\beta$, the Lasso will not perfectly
select all signal variables before the non-signal variables, even when
the signal-to-noise ratio is extremely high \citep{LassoFDR}. As such, the Lasso-based $\bm W$ statistic used in knockoffs never achieves a power of 1; this phenomenon could be remedied by using one of the least-squares-based $\bm W$ mentioned in \cite{knockoff}. The 5-FWER of all
four procedures is again largely unaffected by the coefficient magnitude.


\section{Real data experiment}
\label{sec:real-data-experiment}

In this section, we apply our method to a data set on HIV drug
resistance. Specifically, the data set, described and
analyzed in \cite{Rhee14112006} and also used in the original
knockoffs paper \cite{knockoff}, contains genotype information from
samples of HIV Type 1, along with drug
resistance measurements for 16 drugs across three classes. The
three classes are protease inhibitors
(PI), nucleoside reverse transcriptase inhibitors (NRTI), and
nonnucleoside reverse transcriptase inhibitors (NNRTI), each of which has its own
set of samples. Drug resistance was measured as the log-fold-increase
of resistance as compared to a control, and the genetic information
comes as single nucleotide polymorphisms (SNPs), and thus each binary value
represents the presence or absence of a minor allele at a given locus.

\begin{table}[!htp]
\def~{\hphantom{0}}
\captionsetup{width=1\textwidth, font=normal}
\caption{Multiple testing procedures applied to HIV drug resistance data sets}{%
\begin{tabular}{lcccccccc}
\toprule
Drug & Type & Samples & SNPs & FDR ko & $k$-FWER ko & Step-down &
Step-up & Holm \\
\midrule
APV & PI & 767 & 164 & 19/29 & 10/10 & 14/18 & 14/15 & 14/17 \\
ATV & PI & 328 & 104 & 22/28 & 18/19 & 18/20 & 14/14 & 17/19 \\
IDV & PI & 825 & 165 & 25/42 & 15/17 & 17/21 & 17/20 & 17/20 \\
LPV & PI & 515 & 141 & 17/18 & 13/14 & 17/18 & 13/13 & 14/14 \\
NFV & PI & 842 & 166 & 26/40 & 20/22 & 17/21 & 16/18 & 17/21 \\
RTV & PI & 793 & 163 & 20/26 & 18/18 & 17/23 & 15/17 & 15/20 \\
SQV & PI & 824 & 164 & 20/31 & 19/29 & 16/21 & 15/18 & 15/19 \\
X3TC & NRTI & 629 & 216 & 4/6 & 5/7 & 6/9 & 5/6 & 6/8 \\
ABC & NRTI & 623 & 216 & 16/35 & 16/31 & 8/11 & 8/11 & 8/11 \\
AZT & NRTI & 626 & 216 & 15/21 & 13/17 & 13/21 & 10/14 & 11/18 \\
D4T & NRTI & 625 & 216 & 15/26 & 13/21 & 11/12 & 10/11 & 10/11 \\
DDI & NRTI & 628 & 216 & 2/2 & 5/5 & 8/13 & 7/9 & 8/12 \\
TDF & NRTI & 351 & 148 & 6/6 & 8/8 & 9/11 & 7/8 & 9/10 \\
DLV & NNRTI & 730 & 231 & 10/25 & 10/16 & 11/25 & 11/20 & 11/22 \\
EFV & NNRTI & 732 & 236 & 11/21 & 11/19 & 10/17 & 10/16 & 10/16 \\
NVP & NNRTI & 744 & 236 & 10/23 & 8/13 & 7/15 & 7/12 & 7/13 \\
\midrule
\multicolumn{4}{r}{Average Number of True Discoveries} & \bf{14.9} &
\bf{12.6} & \bf{12.4} & \bf{11.2} & \bf{11.8} \\
\multicolumn{4}{r}{2-FWER} & \bf{0.81} & \bf{0.63} & \bf{0.88} &
\bf{0.63} & \bf{0.81} \\
\bottomrule
\end{tabular}}
\label{HIVtab}
\bigskip
\\
{\small Summary: For each procedure, we report the number of true positives and the
number of total discoveries, separated by a slash. At the end of the
table we report summary statistics for each procedure. ko stands for knockoffs.}\\
\end{table}

In order to analyze the data, some cleaning was required. In
particular, some samples do not have resistance measurements for some
of the drugs, so these samples were removed on a drug-by-drug
basis. Also, some SNPs have so few
mutations that either their effect would be too hard to detect, or
their inclusion actually causes rank-deficiency in the design
matrix. As such, for each drug we only included polymorphisms with at
least five mutations present in the culled sample; this was the
minimum required to ensure all design matrices were full-rank.

We compare our knockoffs procedure to the step-down, step-up, and Holm procedures,
as well as to the original knockoffs procedure for controlling FDR at level $q$. As
$k$-FWER is often used as an exploratory analysis, and to make
analysis comparable with knockoffs for FDR control, we set
$\alpha=0.5$ (FDR controls a mean, and with
$\alpha=0.5$, $k$-FWER controls a median). We set
$k=2$ and $q=0.2$, and ran all five procedures on all 16 drugs, the
results of which are summarized in Table~\ref{HIVtab}.

Although the ground truth is unknown in this case, there
exists an approximate ground truth from treatment-selected mutation (TSM)
panels \citep{Rhee01082005}. These panels list mutations that were
found to be statistically significantly more frequent in virus samples
from individuals treated with a drug in that class than samples from
individuals who had not. Thus in our experiment evaluation, we
consider a SNP discovery for a given drug
to be true if it has a mutation listed in the TSM panel for that drug's class.

The table shows the number of total discoveries and false discoveries made
by each method on each data set. As suspected, FDR-controlling
knockoffs was more powerful than any of the $k$-FWER-controlling
procedures, but is harder to interpret as it never makes a very large
number of discoveries, and thus the FDP may be quite different from
$q$. The remaining procedures 
have varying levels of 2-FWER, but recall that the
error rates reported are likely to be overestimates, as there may be
important SNPs that the TSM panels missed. Still, we see that on this data set, the step-down and Holm
procedures commit more 2-familywise errors than knockoffs, while the
step-up procedure has over 10\% less power than knockoffs.


\section{Discussion}
\label{sec:discussion}
This work leaves a number of important avenues open for future
research. First, we mentioned in Section~\ref{sec:familyw-error-contr}
a number of methods that translate $k$-FWER-controlling procedures
into procedures for controlling the FDX. Investigating the best such
method could yield a powerful method for controlling another important
Type I error rate. Second, \cite{knockoff} mention in passing the
possibility of multiple knockoffs, i.e., constructing $m \ge 1$ sets of
knockoffs and replacing the one-bit $p$-values corresponding to the $\chi_j$'s
with $m+1$-discretized $p$-values. In the setting of FDR control, one
can search over many one-bit $p$-values and need only consider
what fraction, on average, may be false discoveries. However to
control the $k$-FWER, one must keep track of every false discovery,
and we may expect the extra resolution of multiple knockoffs to
provide more power to distinguish true discoveries from false
ones. Lastly, we feel the knockoffs framework is still a largely
untapped resource for generating multiple testing procedures. The
investigation of alternative $W_j$ statistics for ordering variables,
and the extension to other regression settings such as logistic
regression and higher-dimensional problems ($p>n$) are all important
open subjects.

We have presented a novel method for controlling the $k$-FWER in
the context of linear regression. Knockoffs requires no knowledge of
the noise variance and implicitly takes into account the exact
dependence structure of the problem, allowing it to provide
considerable power improvements over state-of-the-art alternatives in
a range of settings. This, along with its intuitive justification and
ease of computation, makes knockoffs a useful practical tool for
multiple hypothesis testing.


{\small
\subsection*{Acknowledgement}
We are grateful to Emmanuel Cand\`es for his encouragements and helpful comments. L.~J. was partially supported by an NIH training grant. W.~S. was partially supported by a General Wang Yaowu Stanford Graduate Fellowship.

\bibliography{ref}
}

\end{document}